\documentclass[final,twocolumn, 10pt]{IEEEtran}
\makeatletter

\let\proof\@undefined
\let\endproof\@undefined
\makeatother
\usepackage[dvips]{graphicx}
\usepackage{amsbsy}
\usepackage{amsmath}
\usepackage{amsthm}
\usepackage{amssymb}
\usepackage{euscript}
\usepackage{citesort}
\usepackage{fancyhdr}
\usepackage{mathrsfs,bbm,dsfont}
\usepackage{amsfonts}
\usepackage[english]{babel}
\usepackage[latin1]{inputenc}
\usepackage{cite}
\input xy
\xyoption{all}
\usepackage{rotating}

\DeclareMathOperator*{\argmin}{argmin}

\DeclareMathOperator*{\tr}{tr}
\DeclareMathOperator*{\Det}{det}

\DeclareMathOperator*{\Ram}{Ram}

\providecommand{\abs}[1]{\ensuremath{\left\lvert #1 \right\rvert}}
\providecommand{\norm}[1]{\ensuremath{\left\Vert #1 \right\Vert}}

\providecommand{\fn}[1]{\ensuremath{\left\Vert #1 \right\Vert}_F^2}
\providecommand{\vv}[1]{\textquotedblleft #1\textquotedblright}
\newcommand{\Q}{\mathbb{Q}}
\newcommand{\Z}{\mathbb{Z}}
\newcommand{\C}{\mathbb{C}}
\newcommand{\R}{\mathbb{R}}
\newcommand{\F}{\mathbb{F}}
\newcommand{\gc}{\mathcal{G}}

\newcommand{\ba}{\bar{\alpha}}
\newcommand{\bt}{\bar{\theta}}
\newcommand{\mo}{\mathcal{O}}
\newcommand{\ma}{\mathcal{A}}

\newcommand{\ms}{\mathcal{S}}

\newcommand{\mh}{\mathcal{H}}

\newcommand{\og}{\overline{\mathcal{G}}}

\newcommand{\e}{\mathbf{e}}
\newcommand{\pg}{\mathfrak{P}}

\DeclareMathOperator*{\as}{as}
\DeclareMathOperator*{\Id}{Id}

\DeclareMathOperator*{\dB}{dB}
\providecommand{\abs}[1]{\ensuremath{\left\lvert #1 \right\rvert}}
\providecommand{\norm}[1]{\ensuremath{\left\Vert #1 \right\Vert}}
\providecommand{\fn}[1]{\ensuremath{\left\Vert #1 \right\Vert}_F^2}
\providecommand{\vv}[1]{\textquotedblleft #1\textquotedblright}

\newtheorem{theo}{Theorem}
\newtheorem{lem}[theo]{Lemma}
\newtheorem{prop}[theo]{Proposition}

\theoremstyle{definition}
\newtheorem{defn}{Definition}
\newtheorem{rem}{Remark}
\newcounter{step}
\newcommand{\step}[1]{\medskip\par\refstepcounter{step}\alph{step})\ \emph{#1}}
\newcounter{stepit}
\newcommand{\stepit}[1]{\medskip\par\refstepcounter{stepit}\alph{stepit})\ \emph{#1}}

\begin{document}
\title{Golden Space-Time Block Coded Modulation}
\author{L. Luzzi $\qquad$ G. Rekaya-Ben Othman $\qquad$ J.-C. Belfiore $\qquad$ E. Viterbo%
\thanks{\scriptsize Jean-Claude Belfiore, Ghaya Rekaya-Ben Othman and Laura Luzzi are with Ecole Nationale Supérieure des Télécommunications (ENST), 46 Rue\ %
Barrault, 75013 Paris, France. E-mail: $\tt \{belfiore,rekaya,luzzi\}$@$\tt enst.fr$.\ %
Emanuele Viterbo is with DEIS -\ %
Universit\`{a} Della Calabria, Via P. Bucci, 42/C, 87036 Rende\ %
(CS), Italy.  E-mail: $\tt viterbo$@$\tt deis.unical.it$.%
}}

\maketitle

\begin{abstract}
In this paper we present a block coded modulation scheme for a $2 \times 2$ MIMO system over slow fading channels, where the inner code is the Golden Code. The scheme is based on a set partitioning of the Golden Code using two-sided ideals whose norm is a power of two. In this case, a lower bound for the minimum determinant is given by the minimum Hamming distance. The description of the ring structure of the quotients suggests further optimization in order to improve the overall distribution of determinants.
Performance simulations show that the GC-RS schemes achieve a significant gain over the uncoded Golden Code.
\end{abstract}


\section{Introduction}
The wide diffusion of wireless communications has led to a growing demand for high-capacity, highly reliable transmission schemes over fading channels. The use of multiple transmit and receive antennas can greatly improve performance because it increases the \emph{diversity order} of the system, defined as the number of independent transmit-receive paths.\\
In order to exploit fully the available diversity, a new class of code designs, called \emph{Space-Time Block Codes}, was developed. In the \emph{coherent}, \emph{block fading} model, where the channel coefficients are supposed to be known at the receiver, and remain constant for a time block, the fundamental criteria for code design are 
\begin{enumerate}
\item[-] the \emph{rank criterion}, stating that the difference of two distinct codewords or \vv{space-time blocks} must be a full-rank matrix,
\item[-] the \emph{determinant criterion}, stating that its minimum determinant ought to be maximized \cite{TSC}.
\end{enumerate}
Codes meeting these two criteria can be constructed using tools from algebraic number theory. In particular, by choosing a subset of a \emph{division algebra} over a number field as our code, we ensure that all the nonzero codewords are invertible. If, furthermore, this subset is contained in an \emph{order} of the algebra, the minimum determinant over all nonzero codewords will be bounded from below and will not vanish when the size of the constellation grows to infinity.\\
In the $2 \times 2$ MIMO case, Belfiore \emph{et al.} \cite{BRV} designed the \emph{Golden Code} $\mathcal{G}$, a full-rate, full-rank and information-lossless code satisfying the non-vanishing determinant condition.
The $n \times n$ MIMO codes that achieve these properties were called \emph{Perfect Codes} in \cite{ORBV} and also studied in \cite{EKPKL}.
\medskip \par
In this paper we focus on the \emph{slow block fading} channel, where the fading coefficients are assumed to be constant for a certain number of time blocks $L$.\footnote{This kind of behaviour might be caused by large obstructions between transmitter and receiver. The model is realistic if $L$ is smaller than the coherence time of the channel; for most practical applications, it has been estimated \cite{BB} that the coherence time is greater than $0.01$ seconds, so that $L<100$ is a legitimate assumption.}\\
Even though fading hinders transmission with respect to the AWGN case, fast fading is actually beneficial because the transmission paths at different times can be regarded as independent. On the contrary, with slow fading the ergodicity assumption must be dropped and the diversity of the system is reduced, leading to a performance loss. \\
This loss can be compensated using \emph{coded modulation}: in a general setting, a full-rank space time block code is used as an \emph{inner code} to guarantee full diversity, and is combined with an \emph{outer code} which improves the minimum determinant. \\
We will take as our inner code the Golden Code: we focus on the problem of designing a \emph{block code} $\{\mathbf{X}=(X_1,\ldots,X_L)\}$, where each component $X_i$ is a Golden codeword. \\
In order to increase the minimum determinant, one can consider the ideals of $\gc$. In \cite{HVB}, Hong \emph{et al.} describe a \emph{set partitioning} of the Golden Code, based on a chain of left ideals $\mathcal{G}_k=\gc B^k$, such that 
the minimum determinant in $\gc_k$ is $2^k$ times that of $\gc$.\\ 
Choosing the components $X_i$ independently in $\gc_k$, one obtains a very simple block code. For small sizes of the signal constellation these subcodes already yield a performance gain with respect to the \vv{uncoded} Golden Code (that is, with respect to choosing $X_i \in \gc$ independently). However, the gain is cancelled out asymptotically by the loss of rate as the size of the signal set grows to infinity, since an energy increase is required to mantain the same spectral efficiency, or bit-rate per channel use. \\
A better performance is achieved when the $X_i$ are not chosen in an independent fashion. In \cite{HVB}, two encoders are combined: a trellis encoder whose output belongs to the quotient $\mathcal{G}_k/\mathcal{G}_{k+1}$, and a lattice encoder for $\mathcal{G}_{k+1}$ (\emph{Trellis Coded Modulation}).
\medskip \par
The global minimum determinant for the block code is given by
\begin{equation*}
\Delta_{\min}= \min_{\mathbf{X} \neq 0} \det \left(\sum_{i=1}^L X_i X_i^H \right) 
\end{equation*}  
This expression is difficult to handle because its \vv{mixed terms} are Frobenius norms of products in $\gc$.
The codes described in \cite{HVB} are designed to maximize the approximate parameter $\Delta_{\min}'=\min_{\mathbf{X} \neq 0} \sum_{i=1}^L \det \left(X_i X_i^H \right)$ and so \emph{a priori} they might be suboptimal; we will here consider the mixed terms and so obtain a tighter bound for $\Delta_{\min}$. \\
A rough estimate of the coding gain for the block code comes from its minimum \vv{Hamming distance}, that is, the minimum number of nonzero components. To increase the Hamming weight, we will take as our outer code an error correcting code over the quotient of $\mathcal{G}$ by one of its ideals.\\
The choice of the ideal must follow some basic requirements. First of all, in order to do a \emph{binary} partitioning, we need to  choose ideals whose index is a power of $2$. Moreover, we will choose \emph{two-sided ideals} to ensure that the quotient group is also a ring.\\
We will describe the ideals of $\gc$ that satisfy our requirements; in particular, we consider the quotient rings $\gc/(1+i)\gc$ and $\gc/2\gc$, which turn out to be isomorphic to the rings of $2 \times 2$ matrices over $\F_2$ and $\F_2[i]$ respectively.\\
Unfortunately, little is known about codes over non-commutative rings, and for the time being we have been unable to exploit the ring structure directly for code construction, except in the simple case of the \emph{repetition code} over the cosets of $(1+i)\gc$. Our performance simulations show that this basic construction can lead to up to $2.9 \dB$ of gain with respect to the \vv{uncoded} case. \\
From the additive point of view the quotient $\gc/2\gc$ is indistinguishable from $\F_{256}$, for which a wide variety of error-correcting codes are available. We can combine a shortened Reed-Solomon code with the encoder of the quotient ring to increase the minimum Hamming distance of the code. \\
Simulation results show that using $4$-QAM constellations, that is using only one lattice point per coset, and with codes of length $L=4$ and $L=6$, we obtain a gain of $6.1\dB$ and $7.0 \dB$ with respect to the uncoded Golden Code at the same spectral efficiency.\\
The construction can be extended to the case of $16$-QAM modulation with multiple points per coset, where the gain is somewhat smaller ($3.9 \dB$ for $L=4$), being limited by the minimum distance in the ideal.
\medskip \par
The paper is organized as follows: in Section \ref{golden_code}, we recall the algebraic construction of the Golden Code and its properties. In Section \ref{gbc}, we describe the general setting for Golden block codes and the coding gain estimates; in Section \ref{asoftgc}, we study the \vv{good ideals} of $\mathcal{G}$ for binary partitioning. In Sections \ref{repetition_code} and \ref{grs} we introduce the repetition code and the Reed-Solomon block code over $\gc$ and discuss their performance obtained through simulations. The interested reader can find in the Appendix the main definitions and theorems concerning quaternion algebras that are cited in the paper. 

\section{The Golden Code} \label{golden_code}

Since we are interested in the partitioning of the Golden Code, we begin by recalling its algebraic construction. For the sake of simplicity, definitions and theorem statements are collected in the Appendix. \\
The Golden Code $\gc$, introduced in \cite{BRV}, is optimal for the case of $2$ transmit and $2$ or more receive antennas. 
This code is constructed using the cyclic division algebra $\mathcal{A}=(\Q(i,\theta)/\Q(i),\sigma,\gamma)$ over the number field $\Q(i,\theta)$, where $\theta=\frac{\sqrt{5}+1}{2}$ is the golden number. The set $\mathcal{A}$ is the $\Q(i,\theta)$-vector space $\Q(i,\theta) \oplus \Q(i,\theta) j$, where $j$ is such that $j^2=\gamma \in \Q(i)^*$, $xj=j\bar{x} \; \forall x \in \Q(i,\theta)$.\\ 
Here we denote by $\sigma$ the canonical conjugacy sending an element $x=a+b\theta \in \Q(i,\theta)$ to $\bar{x}=a+b\bar{\theta}$, where $$\bar{\theta}=1-\theta=\frac{1-\sqrt{5}}{2},\quad \theta\bar{\theta}=-1$$
As its degree over its center $\Q(i)$ is $4$, $\mathcal{A}$ is also called a \emph{quaternion algebra}.\\
If we choose $\gamma=i$, $\gamma$ is not a norm in $\Q(i,\theta)/\Q(i)$ \cite{BRV}, and this implies that $\ma$ is a division algebra (see Theorem 
\ref{norm} in the Appendix).\\
From Theorem \ref{splitting_field}, it follows that $\Q(i,\theta)$ is a splitting field for $\ma$, and so $\mathcal{A}$ is isomorphic to a subalgebra of $\mathcal{M}_2(\Q(i,\theta))$. The inclusion is given by
\begin{equation} \label{j}
x \mapsto \left(\begin{array}{cc} x & 0 \\ 0 & \bar{x} \end{array} \right), \; \forall x \in \Q(i,\theta), \quad j \mapsto \left(\begin{array}{cc} 0 & 1 \\ i & 0 \end{array} \right) 
\end{equation}
That is, every element $X \in \mathcal{A}$ admits a matrix representation 
\begin{equation} \label{A_algebra}
X=\left[\begin{array}{cc} x_1 & x_2 \\ i\bar{x}_2 & \bar{x}_1\end{array}\right],\; x_1, x_2 \in \Q(i,\theta) 
\end{equation}
The Golden Code $\mathcal{G}$ is a subring of $\mathcal{A}$ having two additional properties: the minimum determinant 
$$\delta=\min_{X \neq X', \; X,X' \in \gc} \abs{\Det(X-X')}^2 $$
should be strictly bounded away from $0$, and moreover we want the code to be information lossless.\\ 
For the first condition, if we require that the matrix elements of $X$ belong to the ring of integers $\Z[i,\theta]$ of $\Q(i,\theta)$, then $X$ belongs to the \emph{$\Z[i]$-order}
\begin{equation} \label{O}
\mathcal{O}=\left\{\left[\begin{array}{cc} x_1 & x_2 \\ i\bar{x}_2 & \bar{x}_1\end{array}\right] ,\; x_1, x_2 \in \Z[i,\theta] \right\}
\end{equation}
Since $x \in \Z[i,\theta]$ implies that the reduced norm $N(x)=x\bar{x}$ belongs to $\Z[i]$, we have $\Det(X) \in \Z[i]$,  so $\abs{\Det(X)} \geq 1$ for every $X \in \mathcal{O} \setminus \{0\}$.\\
Each codeword of $\mathcal{O}$ carries two symbols $x_1=a+b\theta$, $x_2=c+d\theta$ in $\Z[i,\theta]$, or equivalently four information symbols $(a,b,c,d) \in \Z[i]^4$: the code is \emph{full-rate}.\\
In order to have an information lossless code, a right principal ideal of $\mathcal{O}$ of the form $\alpha \mo$ was used, where $\alpha=1+i \bar{\theta}$: its matrix representation is 
\begin{equation} \label{A}
A=\left[\begin{array}{cc} \alpha & 0 \\ 0 & \ba \end{array}\right] \in \mathcal{O} 
\end{equation}
The \emph{Golden Code} is defined as $\mathcal{G}=\frac{1}{\sqrt{5}}\alpha \mathcal{O}$. Every codeword in $\gc$ is of the form $X=\frac{1}{\sqrt{5}}AW$, with $W \in \mathcal{O}$:
\begin{equation} \label{golden_codeword}
X=\frac{1}{\sqrt{5}}\left[\begin{array}{cc} \alpha(a +b \theta) & \alpha(c+d \theta) \\ \ba i ( c+d\bt)  & \ba(a +b\bt) \end{array}\right]
\end{equation}
\begin{rem} \label{det_O}
We have seen that $\forall W \in \mathcal{O} \setminus \{0\}$, $\abs{\Det(W)} \geq 1$. Consequently, $\forall X \in \gc \setminus \{0\}$, $\abs{\Det(X)}^2 \geq \delta= \frac{1}{5}$. 
\end{rem}
In fact, if $X=\frac{A}{\sqrt{5}}W$, $\abs{\Det(X)}=\frac{\abs{N(\alpha)}}{5} \abs{\Det(W)}=\abs{\frac{\Det{(W)}}{\sqrt{5}}}$, since $\abs{N(\alpha)}=\abs{2+i}=\sqrt{5}$. \\
The code $\gc$ has \emph{cubic shaping}: it is isometric to the cubic lattice $\Z[i]^4$ (and so it is information lossless). In fact, if we consider the linear mapping $\phi: \mathcal{A} \to \C^4$ that vectorizes matrices
\begin{displaymath}
\phi\left(\left[\begin{array}{cc} a & c \\ b & d \end{array}\right]\right)=(a,b,c,d) \in \C^4,
\end{displaymath} 
then $\phi(\gc)=R\Z[i]^4$, where $R$ is the unitary matrix
\begin{equation}
R=\frac{1}{\sqrt{5}}\left[\begin{array}{rrrr}
\alpha & -\ba i & 0 & 0 \\
0 & 0 & \ba i & \alpha \\
0 & 0 & \alpha & -\ba i \\
\ba & -\alpha i & 0 & 0
\end{array}\right]
\label{R_vectorized}
\end{equation}
Even though $\gc$ is defined as a right ideal, it is easy to see that actually it is a \emph{two-sided ideal}: if $w=w_1+w_2 j \in \mo$, $w_1,w_2 \in \Z[i,\theta]$, 
\begin{displaymath}
\alpha (w_1+ w_2 j)= w_1 \alpha+ w_2 j \ba= (w_1+i\theta w_2j)\alpha,
\end{displaymath}
observing that $\alpha i\theta= i \theta + 1 = \ba$. But 
\begin{equation} \label{two-sided2}
\xi: w_1+w_2 j \mapsto w_1+i\theta w_2 j
\end{equation}
is an homomorphism of $\Z[i]$-modules that maps $\mo$ into itself bijectively, therefore $\alpha \mo=\mo \alpha$.\\
Finally, $\sqrt{5}\gc$ is an \emph{integral ideal} because it is contained in $\mo$. 

\begin{rem}
For the sake of simplicity, in this section we have described the Golden Code as an infinite code. However in a practical transmission scheme, one considers a finite subset of $\gc$, by choosing the information symbols $a,b,c,d$ in a QAM constellation carved from $\Z[i]$.  
\end{rem}

\section{Golden Block Codes} \label{gbc}
We now focus on the case of a \emph{slow block fading} channel, meaning that the channel coefficients remain constant during the transmission of $L$ codewords. The transmitted signal $\mathbf{X}=(X_1,\ldots,X_L)$ will be a vector of Golden codewords in a block code $\mathcal{S} \subset \gc^L$.
The received signal is given by
\begin{equation} 
\mathbf{Y}=H\mathbf{X}+\mathbf{W}, \qquad \mathbf{X}, \mathbf{Y}, \mathbf{W} \in \C^{2 \times 2L},
\end{equation}
where the entries of $H \in \C^{2 \times 2}$ are i.i.d. complex Gaussian random variables with zero mean and variance per real dimension equal to $\frac{1}{2}$, and $\mathbf{W}$ is the complex Gaussian noise with i.i.d. entries of zero mean and variance $N_0$. We consider the coherent case, where the channel matrix $H$ is known at the receiver.\\
The pairwise error probability is bounded by \cite{TSC}
\begin{equation} \label{PEP}
P(\mathbf{X}\mapsto \mathbf{X}') \leq \frac{1}{\left(\sqrt{\Delta_{\min}} \frac{E_{\ms}}{N_0}\right)^4},
\end{equation}
In the above formula, $E_{\ms}$ is the average energy per symbol of $\ms$ and 
\begin{displaymath}
\Delta_{\min}=\min_{\mathbf{X} \in \ms \setminus \{0\}}\; \abs{\Det(\mathbf{X}\mathbf{X}^H)}
\end{displaymath}
In order to minimize the PEP for a given SNR, we should maximize $\Delta_{\min}$. We will show that 
$$\abs{\Det(\mathbf{X}\mathbf{X}^H)}\geq (w_H(\mathbf{X}))^2 \delta,$$
where $w_H(\mathbf{X})$ is the number of nonzero codewords in $(X_1,\ldots,X_L)$ (a sort of \vv{Hamming weight}), and $\delta=\frac{1}{5}$ is the minimum square determinant of the Golden Code.\\ 
Because of the lack of diversity of the channel in the slow fading case, if we simply choose $X_1,\ldots, X_L$ independently in the Golden Code, the code performance will be poor compared to the fast block fading model. We call this scheme the \vv{uncoded Golden Code}: in this case $\Delta_{\min}=\delta$, for any length $L$. \\ 
To compare the error probability of a block codes with that of the uncoded Golden Code of equal length $L$ with the same data rate, we can employ the \emph{asymptotic coding gain} defined in \cite{HVB}:
\begin{equation} \label{acg}
\gamma_{\as}=\frac{\sqrt{\Delta_{\min}}/E_{\mathcal{S}}}{\sqrt{\Delta_{\min,U}}/E_{\mathcal{S},U}},
\end{equation}
where $\Delta_{\min},\Delta_{\min,U}$ and $E_{\mathcal{S}}, E_{\mathcal{S},U}$ are the minimum determinants and average constellation energies of the block code and the uncoded case respectively.\\
In all the cases that we considered, the theoretical gain $\gamma_{\as}$ turned out to be smaller than the actual gain evidenced by computer simulations. This is not surprising, since $\gamma_{\as}$ is only a comparison of the dominant terms in the pairwise error probability.

\subsection{Estimates of the Frobenius norm} \label{Belfiore}
First of all, we give a more explicit expression for $\Det(\mathbf{X}\mathbf{X}^H)$. \\
We define the quaternionic conjugacy in the algebra $\mathcal{A}$: 
\begin{displaymath}
X=\left[\begin{array}{cc} x_1 & x_2 \\ i \bar{x}_2 & \bar{x}_1\end{array}\right] \quad \mapsto \quad \widetilde{X}=\left[\begin{array}{rr} \bar{x}_1 & -x_2 \\ -i \bar{x}_2 & x_1\end{array}\right]
\end{displaymath}
Observe that $\forall X \in \mathcal{A}$,
\begin{align} 
& \widetilde{X} X = \det(X) \mathds{1} \label{prodotto}\\
& \widetilde{X}+ X= (x_1 + \bar{x}_1) \mathds{1}= \tr(X) \mathds{1} \label{somma}\\
& \det(X) = \det({\widetilde{X}}) \label{determinante}
\end{align}
where $\mathds{1}$ denotes the identity matrix.\\
Recall that the \emph{Frobenius norm} of a matrix $M=(m_{i,j})$ is $$\norm{M}_F=\sqrt{\sum_{i,j} \abs{m_{i,j}}^2}$$
Then the following formula holds:
\begin{lem} \label{formula_Belfiore}
$\forall \mathbf{X}=(X_1,\ldots,X_L) \in \mathcal{A}^L$, 
\begin{multline} 
\label{fb}
\Det(\mathbf{X} \mathbf{X}^H)=\det\left(\sum_{i=1}^L X_iX_i^H\right)=\\
=\abs{\Det(X_1)}^2+\ldots+\abs{\Det(X_L)}^2+\sum_{j>i} \norm{\widetilde{X}_j X_i}_F^2
\end{multline}
\end{lem}
The proof can be found in Appendix \ref{proofs}.\\
We also state some simple properties of the quaternionic conjugate and of the Frobenius norm that will be useful in the sequel:
\begin{rem} \label{summary}
\begin{enumerate}
\item[a)] If $W \in \mathcal{O}$, $\norm{W}_F^2 \in \Z$. 
\item[b)] Let $X, Y$ be two $2 \times 2$ complex-valued matrices. Then
\begin{equation}
\begin{split} \label{det}
&\fn{X} \geq 2\abs{\Det(X)}, \\
&\fn{\widetilde{X}Y}\geq 2\abs{\Det(X)}\abs{\Det(Y)}
\end{split} 
\end{equation}
In particular $\forall W \in \mathcal{O}\setminus\{0\}$,
\begin{equation} 
\fn{W}\geq 2\abs{\Det(W)} \geq 2 \label{estimate_O}
\end{equation}
\item[c)] If $X_1,X_2 \in \mathcal{G} \setminus \{0\}$, 
\begin{equation}
\norm{\widetilde{X}_2X_1}_F^2 \geq \frac{2}{5}=2\delta \label{product_norm}
\end{equation} 
\end{enumerate}
\end{rem}

From equation (\ref{det}), it follows that the determinant is bounded from below by the squared Hamming weight:
\begin{lem} \label{delta_d}
Let $\mathbf{X}=(X_1,\ldots,X_L) \in \gc^L$. Then
\begin{displaymath}
\Det(\mathbf{X} \mathbf{X}^H)\geq \left(\sum_{i=1}^L \abs{\Det(X_i)} \right)^2 \geq (w_H(\mathbf{X}))^2\delta,
\end{displaymath}
where $w_H(\mathbf{X})=\# \{i \in \{1,\ldots,L\} \,|\, X_i \neq 0\}$ is the Hamming weight of the block $\mathbf{X}$.
\end{lem}

\section{Two-sided ideals of $\gc$} \label{asoftgc}
The choice of a good block code of length $L$ will be based on a partition chain of ideals of the Golden Code. We would like to obtain a binary partition, which is simpler to use for coding and fully compatible with the choice of a QAM constellation: we must then use ideals whose index is a power of $2$, that is, whose norm is a power of $1+i$.\\
A similar construction appears in \cite{HVB} and employs one-sided ideals. However, in order to have good estimates of the coding gain, because of the mixed terms in the minimum determinant formula (\ref{fb}), we need to take the ring structure into account: we will choose \emph{two-sided ideals} to ensure that the ideals are invariant with respect to the quaternionic conjugacy and multiplication on both sides, and that the quotient group is also a ring.
\medskip \par
In this section we describe the structure of the two-sided ideals of $\gc$ whose norm is a power of $1+i$. Unfortunately, we will see that the only two-sided ideals with this property are the trivial ones. We then study the corresponding quotient rings, which are rings of matrices over non-integral rings. \\
For these constructions we will need some notions from non-commutative algebra (see Appendix \ref{ivmo}), relating the existence of two-sided ideals to the ramification of primes over the base field. We will also show that $\mo$ is a maximal order of $\ma$.

As we have seen in Section \ref{golden_code}, $\mo=\Z[i,\theta] \oplus \Z[i,\theta] j$ is a $\Z[i]$-order of $\ma$, and $\og=\sqrt{5}\gc=\alpha \mo$ is a two-sided principal ideal of $\mo$.\\
$\sqrt{5}\gc$ is also a prime ideal since $\sqrt{5}\gc \cap \Z[i]=(2+i)$ is a prime ideal of $\Z[i]$ (see Theorem \ref{1-1} in the Appendix).\\
Observe that the prime ideals $(2+i)$ and $(2-i)$ of $\Z[i]$ are both ramified in $\ma$: in fact 
\begin{displaymath}
(2+i)=(\alpha)^2, \text{ and } (2-i)=(\alpha')^2, \text{ where } \alpha'=1-i\bt
\end{displaymath} 
(Remark that $\alpha=i\theta \bar{\alpha}$, $\alpha'=-i\bt \bar{\alpha'}$).
 
\begin{prop}
$\mo$ is a maximal order.
\end{prop}

\begin{proof} 
$\ma$ is a quaternion algebra unramified at infinity: the infinite primes are complex (because the base field $\Q(i)$ is imaginary quadratic) and they can't be ramified. Then one can check that $\mo$ is maximal through the computation of its reduced discriminant $d(\mathcal{O})$ (see Proposition \ref{discriminant} in the Appendix).\\
$d(\mathcal{O})$ is equal to $\sqrt{\abs{\Det(\tr(w_k w_l))}}\Z[i]$, where $\{w_1=1$, $w_2=\theta$,$w_3=j$, $w_4=\theta j\}$ is the basis of $\mo$ over $\Z[i]$:  

\begin{align*}
&(w_k w_l)_{1 \leq k,l \leq 4}=\left(\begin{array}{rrrr}
1 & \theta & j & \theta j \\
\theta & \theta^2 & \theta j & \theta^2 j \\
j & \bt j & i & i\bt \\
\theta j & -j & \theta i & -i 
\end{array}\right), \\
&\Det(\tr(w_i w_j))= \Det\left( \begin{array}{rrrr}
2 & 1 & 0 & 0 \\
1 & 3 & 0 & 0 \\
0 & 0 & 2i & i \\
0 & 0 & i & -2i 
\end{array} \right)=25
\end{align*}

Then $d(\mo)=5\Z[i]$. If $\mo$ were strictly contained in a maximal order $\mo'$, $d(\mo')$ would be strictly larger than $5\Z[i]$. But we know from Proposition \ref{discriminant} that $d(\mo')$ is the product of all ramified primes of $\ma$; in particular it should be contained in the ideals $(2+i)$ and $(2-i)$. But then it would be contained in $5\Z[i]$, a contradiction. Then $\mo$ is a maximal order, and $\gc$ is a normal ideal.
\end{proof}

Since $\mo$ is maximal, from Proposition \ref{discriminant} we also learn that $(2+i)$ and $(2-i)$ are the only ramified primes in $\ma$. \\
Then Theorem \ref{two-sided} implies that the prime two-sided ideals of $\mo$ are either of the form $p \mo$, where $p$ is prime in $\Z[i]$, or belong to $\{\alpha \mo, \alpha' \mo\}$. \\
It follows that the only two-sided ideals of $\gc$ whose norm is a power of $1+i$ are the trivial ideals of the form $(1+i)^k\gc$. 

\subsection{The quotient ring $\og/(1+i) \og$}
In the sequel, we will denote by $\og$ the integral ideal $\sqrt{5}\gc$.\\
Consider the prime ideal $(1+i) \mo$. $\og$ and $(1+i)\mo$ are \emph{coprime} ideals, that is $\og+(1+i)\mo=\mo$; as a consequence, $\og \cap (1+i)\mo=\og(1+i)\mo=(1+i)\og$. Recall the following basic result:

\begin{theo}[third isomorphism theorem for rings]
Let $I$ and $J$ be ideals in a ring $R$. Then $\frac{I}{I \cap J} \cong \frac{I+J}{J}$.
\end{theo}

If $I=\og$ and $J=(1+i)\mo$, we get 
\begin{equation} \label{isomorphism}
\frac{\og}{(1+i) \og}\cong\frac{\mo}{(1+i)\mo} 
\end{equation} 

If $\pi_{\og}:\og \to \og/(1+i) \og$ and $\pi_{\mo}: \mo \to \mo/(1+i)\mo$ are the canonical projections on the quotient, the ring isomorphism in (\ref{isomorphism}) is simply given by $\pi_{\og}(g) \mapsto \pi_{\mo}(g)$.\\  
Theorem \ref{1-1} implies that $\mo/(1+i)\mo$ is a simple algebra over $\Z[i]/(1+i)\cong \F_2$. 
We denote the image of $x \in \mo$ through $\pi_{\mo}$ with $[x]$. 

\begin{lem} \label{1+i}
$\mo/(1+i)\mo$ is isomorphic to the ring $\mathcal{M}_2(\F_2)$ of $2 \times 2$ matrices over $\F_2$. 
\end{lem}

\begin{proof}
We use the well-known lemma \cite{Hu}:
\begin{lem} \label{free}
Let $R$ be a ring with identity, $I$ a proper ideal of $R$, $M$ a free $R$-module with basis $X$ and $\pi: M \to M/IM$ the canonical projection. Then $M/IM$ is a free $R/I$-module with basis $\pi(X)$ and $\abs{\pi(X)}=\abs{X}$.
\end{lem}
We know that $\mo/(1+i)\mo$ is a $\Z[i]$-module; the lemma implies that it is also a free $\Z[i]/(1+i)$-module, that is a vector space over $\F_2$, whose basis is $\{[1],[\theta],[j],[\theta j]\}$.\\
We define an homomorphism of $\F_2$-vector spaces $\psi: \protect{\mo/(1+i)\mo} \to \mathcal{M}_2(\F_2)$ by specifying the images of the basis:
\begin{align*} 
&\psi([1])=\mathds{1}, \quad \psi([\theta])=\begin{pmatrix} 0 & 1 \\ 1 & 1 \end{pmatrix}, \\
&\psi([j])=\begin{pmatrix} 0 & 1 \\ 1 & 0 \end{pmatrix}, \quad \psi([\theta j])=\psi([\theta]) \psi([j]) 
\end{align*} 
It is one-to-one since $\psi([1]), \psi([\theta]), \psi([j]), \psi([\theta j])$ are linearly independent. To prove that $\psi$ is also a ring homomorphism, it is sufficient to verify that $\psi(w_i w_j)=\psi(w_i) \psi(w_j)$ for all pairs of basis vectors $w_i, w_j$. 
\end{proof}

Recall that as a $\Z[i]$-lattice, $\og$ is isometric to $\sqrt{5}\Z[i]^4$, and a canonical basis is given by $\{\alpha,\alpha\theta, \alpha j,  \alpha \theta j\}$.
The corresponding elements $\psi([\alpha]),\psi([\alpha\theta]),\psi([\alpha j]),\psi[\alpha \theta j])$ of $\mathcal{M}_2(\F_2)$ are
\begin{equation}
\begin{split} 
\e_1=\begin{pmatrix} 0 & 1 \\ 1 & 1 \end{pmatrix}, \quad  \e_2=\begin{pmatrix} 1 & 1 \\ 1 & 0 \end{pmatrix},  \\ 
\e_3=\begin{pmatrix} 1 & 0 \\ 1 & 1 \end{pmatrix}, \quad  \e_4=\begin{pmatrix} 1 & 1 \\ 0 & 1 \end{pmatrix}.  \label{basis}
\end{split}
\end{equation}
It is easy to check that the only invertible elements in $\mathcal{M}_2(\F_2)$ are
\begin{displaymath}
\e_1,\;\e_2,\;\e_3,\;\e_4,\;\e_1+\e_2=\mathds{1},\;\e_3+\e_4=\varphi(j)
\end{displaymath}

Observe that the lifts to $\gc$ of non-invertible elements have a higher determinant:
\begin{rem} \label{non-invertible}
If $M \in \mathcal{M}_2(\F_2) \setminus \{0\}$ is non-invertible, $$\min_{X \in \gc,\; \pi_{\og}(\sqrt{5}X)=M} \abs{\det(X)}^2 \geq 2\delta$$ 
\end{rem}
\begin{proof}
$\pi_{\og}(X)$ is non-invertible in $\og/(1+i)\og$ if and only if its determinant is non-invertible in $\Z[i]/(1+i)$, that is, $\det(X)=\widetilde{X} X \in (1+i)\setminus \{0\}$. (If $M \neq 0$, $\det(X) \neq 0$, since $\ma$ is a division ring.) \\
Then $\abs{\det(\widetilde{X} X)}= \abs{\det(X)}^2 \geq 2\delta$.  
\end{proof}

\subsection{The quotient ring $\og/ 2\og$} \label{quotient2}
Again, $\og$ and $2 \mo$ are coprime and so $\og+2\mo=\mo$, $\og\cap \mo=2\og$; from the third isomorphism theorem for rings, $\frac{\og}{2\og} \cong \frac{\mo}{2\mo}$.  

\begin{lem} \label{quotient_lemma}
$\mo/2\mo$ is isomorphic to the ring $\mathcal{M}_2(\F_2[i])$ of $2 \times 2$ matrices over the ring $\F_2[i]$. 
\end{lem}

\begin{proof}
First of all, Lemma \ref{free} implies that $\mo/2 \mo$ is a free $\Z[i]/2$-module, that is a free $\F_2[i]$-module, of dimension $4$. 
As in the previous case, we can construct an explicit homomorphism of $\F_2[i]$-modules $\phi: \mo/2\mo \to \mathcal{M}_2(\F_2[i])$:
\begin{align*}
\phi([1])=\mathds{1}, \quad \phi([\theta])=\begin{pmatrix} 1+i & 1 \\ i & i \end{pmatrix},\\
\phi([j])=\begin{pmatrix} 0 & 1 \\ i & 0 \end{pmatrix}, \quad \phi([\theta j])=\phi([\theta]) \phi([j]) 
\end{align*} 
One can easily check that $\phi$ is bijective (the images of the basis elements being linearly independent) and that it is a ring homomorphism.
\end{proof}

To find an explicit isomorphism between $\og/2\og$ and $\mathcal{M}_2(\F_2)$, consider the following diagram, where $\pi_{\og}: \og \to \og/2\og$ is the projection on the quotient, $\varphi$ is given by the third isomorphism theorem for rings, and $\phi: \mo/2\mo \to \mathcal{M}_2(\F_2[i])$ is the mapping defined in Lemma \ref{quotient_lemma}: 
\begin{equation*} 
\og \xrightarrow{\quad\pi_{\og}\quad} \og/2\og \xrightarrow{\quad\varphi\quad} \mo/2\mo \xrightarrow{\quad\phi\quad} \mathcal{M}_2(\F_2[i])
\end{equation*}
The basis $\{\alpha,\alpha\theta,\alpha j,\alpha \theta j\}$ of $\og$ as a $\Z[i]$-module is also a basis of $\og/2\og$ as an $\F_2[i]$-module. The isomorphism $\varphi$ is simply the composition of the inclusion $\og \hookrightarrow \mo$ and the quotient mod $(1+i)\mo$.
We can compute the images through $\phi$ of the basis vectors: observing that
\begin{align*}
&\alpha=1+i-i\theta,\quad \alpha \theta=\theta-i,\\
&\alpha j=(1+i-i\theta)j,\quad\alpha\theta j=(\theta-i)j,
\end{align*}
we get
\begin{align} \label{basis2}
&\phi(\alpha)=
\begin{pmatrix} 0 & i \\ 1 & i\end{pmatrix}, \quad \phi(\alpha \theta)=
\begin{pmatrix} 1 & 1 \\ i & 0\end{pmatrix}, \\
&\phi(\alpha j)=
\begin{pmatrix}1 & 0 \\ 1 & 1 \end{pmatrix}, \quad \phi(\alpha \theta j)=
\begin{pmatrix}i & 1 \\ 0 & i \end{pmatrix}.
\end{align}
Also in this case, the lifts $X$ of non-invertible elements of $\mathcal{M}_2(\F_2[i])$ in $\gc$ will have non-invertible determinant, that is $\abs{\det{(X)}}^2\geq 2$.

\subsection{The encoder}
The codes that we consider follow the general outline of Forney's \emph{coset codes}, taking advantage of the decomposition $\gc=[\gc/I]+I$, where $I$ is $(1+i)\gc$ or $2\gc$, and $[\mathcal{G}/I]$ denotes a set of coset leaders.

\begin{enumerate}
\item[-] a binary $(n,k,d_{\min})$ encoder operates on some of the information data, and these coded bits are used to select $(C_1,\ldots,C_L) \in (\gc/I)^L$. 
\item[-] the remaining information bits are left uncoded and used to select $(Z_1,\ldots,Z_L) \in I^L$.
\item[-] the corresponding block codeword is $\mathbf{X}=(c_1+Z_1,\ldots,c_L+Z_L) \in \gc^L$, where $c_i$ is the coset leader of $C_i$.
\end{enumerate} 
The encoder is illustrated in Figure \ref{encoder}.\\
For a coset code, $\Delta_{\min}$ is bounded by the minimum determinant of $I$ and the minimum distance $d_{\min}$ of the binary code:
\begin{equation} \label{bound}
\Delta_{\min} \geq \min\left(\min_{X \in I \setminus \{0\}} \abs{\Det(X)}^2,d_{\min}^2\delta\right)
\end{equation}
In fact, if $(c_1,\ldots,c_L)=\mathbf{0}$, then $\mathbf{X} \in I^L$, and for $\mathbf{X} \neq \mathbf{0}$, $\det(\mathbf{X}\mathbf{X}^H) \geq \min_{X \in I \setminus \{0\}} \abs{\Det(X)}^2$. If on the contrary $(c_1,\ldots,c_L) \neq \mathbf{0}$, there are at least $d_{\min}$ components of $\mathbf{X}$ which do not belong to $I$, and consequently are nonzero, and $\Det(\mathbf{X}\mathbf{X}^H) \geq \delta w_H(\mathbf{X})\geq \delta d_{\min}^2$.\\
So the performance of a coset code will be always limited by the minimum determinant of $I$, except if the code on $I^L$ is the zero code. \\
If $I$ is simply $(1+i)\gc$ or $2\gc$, the set of possible coordinates $(a,b,c,d)$ for the coset leaders of $I$ in $\gc$ \emph{coincides} with the (BPSK)$^4$ and ($4$-QAM)$^4$ constellations respectively. This makes it much easier to implement coset codes with high Hamming distance.

\begin{figure}[btp]
\begin{center}
\includegraphics[width=0.35\textwidth]{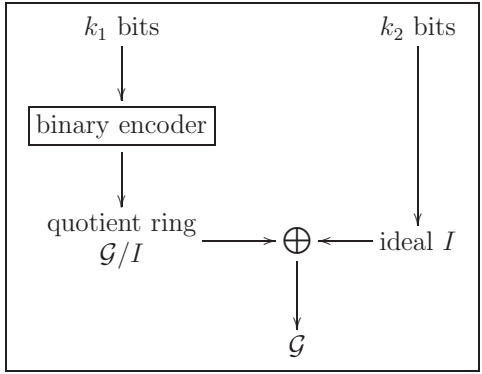}
\end{center}
\caption{The general structure of the encoder.}
\label{encoder}
\end{figure}

\section{The repetition code} \label{repetition_code}
Here we consider the case where $I=(1+i)\gc$, and the binary code is simply the repetition code of length $2$ over $\gc/I$.  
If $\pi: \gc \to \og/(1+i)\og$ is the projection on the quotient ring ($\pi(X)=\pi_{\og}(\sqrt{5}X)$), we define
$$\mathcal{C}=\{\mathbf{X}=(X_1,X_2) \in \gc^2 \;|\; \pi(X_1)=\pi(X_2)\}$$

\subsection{The minimum determinant}
Recall that as we have seen in Lemma \ref{formula_Belfiore}, 
$$\Det(\mathbf{X}\mathbf{X}^H)=\abs{\Det(X_1)}^2+\abs{\Det(X_2)}^2+\norm{\widetilde{X}_2 X_1}_F^2$$
With the code $\mathcal{C}$, we have $\Delta_{\min}=4\delta$. In fact if $(X_1,0)$ (respectively, $(0,X_2)$) is a codeword of Hamming weight $1$, clearly $\pi(X_1)=0$ and $\Det(\mathbf{X}\mathbf{X}^H)=\abs{\Det(X_1)}^2$ is greater than the minimum square determinant in $(1+i)\gc$, which is $4\delta$. If on the contrary $\pi(X_1)=\pi(X_2) \neq 0$, $$\Det(\mathbf{X}\mathbf{X}^H)\geq \left(\abs{\Det(X_1)}+\abs{\Det(X_2)}\right)^2\geq 4\delta$$ because of equation (\ref{det}). \\
By choosing any bijection $h$ of the quotient ring $\og/(1+i)\og$ in itself, one obtains a simple variation of the repetition scheme: 
$$\mathcal{C}_h=\{\mathbf{X}=(X_1,X_2) \in \gc^2 \;|\; \pi(X_2)=h(\pi(X_1))\}$$ 

\begin{rem}
A suitable choice of $h$ can slightly improve performance. In the case of the repetition code, suppose that $\pi(X_1)=\pi(X_2)=C_i$.
\begin{enumerate}
\item[-] If $C_i$ is invertible in $\mathcal{M}_2(\F_2)$, then
$\widetilde{C}_iC_i=\Det(C_i)\mathds{1}=\mathds{1}=\e_1+\e_2$ 
in the basis (\ref{basis}), and so the minimum determinant of a codeword $\widetilde{X}_2X_1 \in \pi^{-1}(\widetilde{C}_iC_i)$ is also $1$, and the minimum of $\norm{\widetilde{X}_2X_1}_F^2$ is $2\delta$. Thus $\Det(\mathbf{X}\mathbf{X}^H)\geq (1+1+2)\delta=4\delta$.
\item[-] If on the other side $C_i$ corresponds to a non-invertible, nonzero element in $\mathcal{M}_2(\F_2)$, then (see Remark \ref{non-invertible}) 
$$\min_{X \in \pi^{-1}(C_i)}\abs{\Det(X)}\geq \sqrt{2\delta}$$ 
and $\Det(\mathbf{X}\mathbf{X}^H)\geq \left(\abs{\Det(X_1)}+\abs{\Det(X_2)}\right)^2\geq(2\sqrt{2\delta})^2= 8\delta$. 
\end{enumerate}
This remark suggests that it might be more convenient to consider a group homomorphism $h: \mathcal{M}_2(\F_2) \to \mathcal{M}_2(\F_2)$ which maps invertible elements into non-invertible elements, raising the minimum determinant to $6\delta$ if $C_i$ invertible, $h(C_i)$ non-invertible: $\norm{\widetilde{X}_2X_1}_F^2\geq 2\sqrt{2}\delta$, but $\norm{\widetilde{X}_2X_1}_F^2 \in \delta\Z$ (see Remark \ref{summary}) and so $\norm{\widetilde{X}_2X_1}_F^2\geq 3\delta$, and $\Det(\mathbf{X}\mathbf{X}^H)\geq (1+2+3)\delta=6\delta$.\\ 
Such a function $\bar{h}$ is not difficult to define, and in the case of $4-QAM$ modulation, an exhaustive search on the finite lattice shows that the distribution of determinants for $\mathcal{C}_{\bar{h}}$ is indeed better.\footnote{In fact, if we define $\bar{h}(\e_1)=\e_1+\e_2+\e_4$, $\bar{h}(\e_2)=\e_2+\e_3+\e_4$, $\bar{h}(\e_3)=\e_1+\e_2+\e_3$, $\bar{h}(\e_4)=\e_1+\e_3+\e_4$ with respect to the basis (\ref{basis}), we have
\begin{scriptsize}
\begin{align*}
&\sum_{\mathbf{X} \in \mathcal{C}} q^{Det(\mathbf{X} \mathbf{X}^H)}=1+66q^4+120q^8+48q^{10}+202q^{16}+	\ldots\\
&\sum_{\mathbf{X} \in \mathcal{C}_{\bar{h}}} q^{Det(\mathbf{X} \mathbf{X}^H)}=1+24q^4+61q^8+24q^9+8q^{10}+74q^{12}+\ldots
\end{align*} 
\end{scriptsize} 
}
\end{rem}
\subsection{The encoder}
Only $4$ bits are needed to select an element of $\og/(1+i)\og\cong \mathcal{M}_2(\F_2)$, while the number of bits needed to select an element in the ideal depends on the chosen modulation scheme. Using $4$-QAM constellations, the two choices of an element in $\protect{(1+i)\gc}$ require $4$ bits each: in total, each codeword carries $12$ information bits, yielding a spectral efficiency of $3$ bpcu.\\
Suppose that $(b_1,\ldots,b_{12})$ is the binary input:
\begin{enumerate}
\item[-] $(b_1,\ldots,b_4)$ are used to select the matrix $b_1\e_1+b_2\e_2+b_3\e_3+b_4\e_4 \in \mathcal{M}_2(\F_2)$ in the basis (\ref{basis}). The corresponding element of $[\gc/(1+i)\gc]$ is $C=[b_1\alpha+b_2\alpha\theta+b_3\alpha j+b_4\alpha \theta j]$. 
\item[-] $(b_5,\ldots,b_{12}$ are used to select two codewords in $\protect{(1+i)\gc}$: $X_1=(1+i)(b_5\alpha+b_6\alpha\theta+b_7\alpha j+b_8\alpha \theta j)$, $X_2=\protect{(1+i)}(b_{9}\alpha+b_{10}\alpha\theta+b_{11}\alpha j+b_{12}\alpha \theta j)$.
\item[-] The final block codeword is $(C+X_1, h(C)+X_2)$. 
\end{enumerate} 

\subsection{Asymptotic coding gain} 
Since the minimum determinant doesn't change, the asymptotic coding gain estimate is the same for all choices of $h$. \\
We compare these schemes with the uncoded Golden Code at $3$ bpcu, using $4$-QAM constellations for the symbols $a,c$ and BPSK constellations for the symbols $b,d$ in each Golden codeword (see equation \ref{golden_codeword}). The average energy per symbol is $E_{\mathcal{S}}=0.5(0.5+0.25)=0.375$, and
$$\gamma_{\as}=\frac{\sqrt{\Delta_{\min}}/E_{\mathcal{S}}}{\sqrt{\Delta_{\min,U}}/E_{\mathcal{S},U}}=\frac{2/0.5}{1/0.375}=1.5,$$
This computation gives a theoretical gain of at least $10\log_{10}(1.5)\dB=1.7 \dB$.  

\subsection*{Simulation results}
Figure \ref{repetition} shows the performance of the codes $\mathcal{C}_{\Id}$ and $\mathcal{C}_{\bar{h}}$, which gain $2.4 \dB$ and $2.9 \dB$ respectively over the uncoded scheme at $3$ bpcu at the frame error rate of $10^{-3}$, supposing that the channel is constant for $2$ time blocks. 

\begin{figure}[tbp] 
\begin{center}
\includegraphics[width=0.4\textwidth]{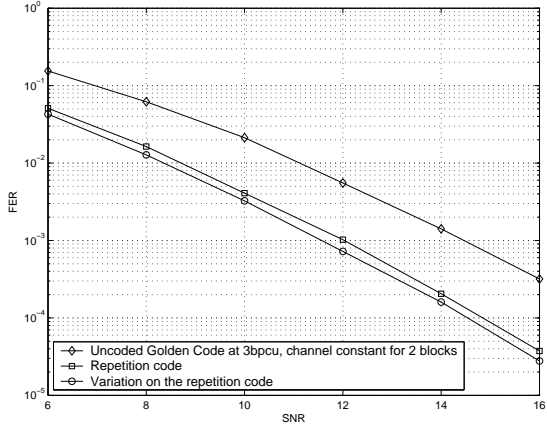}
\caption{Performance of the repetition code $\mathcal{C}_{\Id}$ and of the variation $\mathcal{C}_{\bar{h}}$ at $3$ bpcu compared with the uncoded Golden Code scheme with the same spectral efficiency. The channel is supposed to be constant for $2$ time blocks.}
\label{repetition}
\end{center}
\end{figure}

\section{Golden Reed-Solomon Codes} \label{grs}
The repetition code has the advantage of simplicity, but clearly its performance is limited by the fact that the minimum Hamming distance is only $1$. To increase the Hamming distance, we need to use a more sophisticated error-correcting code.\\
As we have seen in the previous sections, in addition to the minimum Hamming distance, also the multiplicative structure and the minimum number of non-invertible components have a significant influence on the coding gain of a block code design. Ideally, in order to keep track of these pa\-ra\-me\-ters, 
one ought to employ error-correcting codes on $\mathcal{M}_2(\F_2[i])$. However, at present very little is known about codes over non-commutative rings;
we choose shortened Reed-Solomon codes instead because they are maximum distance separable and their implementation is very simple; we will restrict our attention to the additive structure, defining a group isomorphism between $\og/2\og$ and the finite field $\F_{256}.$

\subsection{The 4-QAM case} \label{4Q}
Using $4$-QAM constellations to modulate each of the $4$ symbols $a,b,c,d$ in a Golden codeword (\ref{golden_codeword}), we obtain a total of $256$ codewords, one in each coset of $2\gc$. 
\medskip \par
We consider an $(n,k,d_{\min})$ Reed-Solomon code over $\F_{256}$. 
Each quadruple $(a,b,c,d)$ of $4$-QAM signals carries $8$ bits or one byte; each block of $n$ Golden codewords will carry $n$ bytes, corresponding to $k$ information bytes.\\
The encoding procedure involves several steps:

\step{Reed-Solomon encoding:}\label{s2}\\%
Each information byte can be seen as a binary polynomial of degree $\leq 8$, that is, an element of the Galois Field $\F_{256}$. An information message of $k$ bytes, seen as a vector $\mathbf{U}=(U_1,\ldots,U_k) \in \F_{256}^k$, is encoded into a codeword $\mathbf{V}=(V_1,\ldots,V_n)\in \F_{256}^n$ using the RS$(n,k,d_{\min})$ shortened code $\mathcal{C}$. For our purposes, it is much better to use a \emph{systematic} version of the code that preserves the first $k$ bits of the input.

\step{From the Galois field $\F_{256}$ to the matrix ring $\mathcal{M}_2(\F_2[i])$:} \label{3}\\%
We can represent the elements of $\mathcal{M}_2(\F_2[i])$ as bytes, simply by vectorising each matrix and separating real and imaginary parts.
Since we are only working with the additive structure, we can identify $\F_{256}$ and $\mathcal{M}_2(\F_2[i])$, which are both $\F_2$-vector spaces of dimension $8$. According to our simulation results, it seems that the choice of the linear identification has very little influence on the code performance.
\step{From the matrix ring $\mathcal{M}_2(\F_2[i])$ to the quotient ring $\og/2\og$:}\\ \label{4}%
For this step we make use of the isomorphism of $\F_2[i]$-modules $(\varphi \circ \phi)^{-1}: \mathcal{M}_2(\F_2[i]) \to \og/2\og$ described in Section \ref{quotient2}that relates the coordinates with respect to the bases $\mathcal{B}_{\og}=\{\alpha, \alpha \theta, \alpha j, \alpha \theta j\}$ and (\ref{basis2}).
Let $(a,b,c,d) \in \Z_2[i]^4$ be the coordinates of a codeword in the basis $\mathcal{B}_{\og}$.
\step{Golden Code encoding}: \\%
For each of the $n$ vector components, the symbols $a$,$b$,$c$,$d \in \Z_2[i]$  correspond to four $4$-QAM signals, and can be encoded into a Golden codeword of the form (\ref{golden_codeword}). Thus we have obtained a Golden block $\mathbf{X}=(X_1,X_2,\ldots,X_n)=\xi(\mathbf{V})$, where $\xi: \F_{256}^n \to \gc^n$ is injective.
\subsection{Decoding} \label{decoding}
ML decoding consists in the search for the minimum of the Euclidean distance
$$ \sum_{i=1}^n\norm{HX_i-Y_i}^2 $$
over all the images $\mathbf{X}=\xi(\mathbf{V'})$ of Reed-Solomon codewords.\\
One can first compute and store in memory the Euclidean distances 
\begin{equation} \label{dij}
d(i,j)=\norm{HX^{(j)}-Y_i}^2
\end{equation}
for every component $i=1,\ldots,n$ of the received vector $\mathbf{Y}$ and for all the Golden codewords $X^{(j)}, j=0,..,255$ that can be obtained from a quadruple $U^{(j)}$ of $4$-QAM symbols. \\
The search for the minimum can be carried out using the Viterbi algorithm or a tree search algorithm.
\subsubsection{Stack decoding} \label{stack}
For our computer simulations, we have chosen to use a stack decoding algorithm.
If the code is based on an $(n,k,d_{\min})$ Reed-Solomon code with systematic generator matrix, the $(256)^k$ codewords are the possible paths in a full tree with height $k$ and $256$ outgoing branches per node.\\
The decoder will store in a stack a certain number of triples $(s,\mathbf{u},d_{\mathbf{u}})$, where $\mathbf{u}$ is an incomplete path of length $s$ in the tree, and $d_{\mathbf{u}}$ is its distance from the initial segment $(Y_1,\ldots,Y_s)$ of $\mathbf{Y}$.\\
An upper bound $T$ for the minimum distance of the received point to the lattice of Golden-RS codewords will be used as a \vv{cost function} for the stack.
\stepit{Sorting of distances:} Before the search, for each component $i$, the distances $\{d(i,j)\}_{j=0,..,255}$ of equation (\ref{dij}) are sorted in increasing order: let
$$d(i,j_1(i)),d(i,j_2(i)),\ldots, d(i,j_{256}(i))$$ be the resulting sequence.
\stepit{First step:} At the beginning, the initial segments of length $1$ are inserted into a previously empty stack: the triples $$(1,j_1(0),d(0,j_1(0))),\ldots,(1,j_{256}(0),d(0,j_{256}(0)))$$ are entered in decreasing order with respect to the distance, discarding those whose distances are greater than $T$.
\stepit{Intermediate steps:} At each iteration of the algorithm, the triple $(s,\mathbf{u}=(j^{(1)},\ldots,j^{(s)}),d_{\mathbf{u}})$ currently at the top of the stack is examined. 
\begin{itemize}
\item If $s<k$, its \vv{children} nodes 
\begin{align*}
&(s, (\mathbf{u},r)=(j^{(1)},\ldots,j^{(s)},r),d_{(\mathbf{u},r)}), \\ &\text{for } r=j_1(s+1),j_2(s+1),\ldots,j_{256}(s+1)
\end{align*}
are generated, updating the corresponding Euclidean distances:
$$d_{(\mathbf{u},r)}=d_{\mathbf{u}}+d(s+1,r)$$
The \vv{parent} node is deleted from the stack and the children are inserted in the stack and sorted with respect to distance, or discarded if the distance is greater than $T$.\\
(Remark that since you know the minimum distances component-wise, you can require a stronger condition without losing optimality, namely, $d_{(u,r)}+\sum_{t=s+1}^n d(t,j_1(t))<T$).
\item If $s=k$, generate the Reed-Solomon codeword $\mathbf{v}=(v_1,\ldots,v_n)=G\mathbf{u}$ and store $(n,\mathbf{v},d_{\mathbf{v}})$ in the stack (recall that $\mathbf{u}$ is an initial segment of $\mathbf{v}$), where
$$d_{\mathbf{v}}=d_{\mathbf{u}}+\sum_{t=k+1}^{n} d(t,v_t)$$
\item If $s=n$, the search terminates and the initial segment of length $k$ of $\mathbf{u}$ is the decoded message.
\end{itemize}
\stepit{Choice of the cost function $T$:} A simple bound  for the decoder may be the distance from the received signal of the (unique) Golden-RS codeword corresponding to the \vv{closest choice} $\left(U^{(j_0(1))},\ldots, U^{(j_0(k))}\right)$ for the first $k$ components. Any subset of $k$ components may be used as well to improve the minimum provided that the corresponding lines in the Reed-Solomon generator matrix are linearly independent.

\subsection{Simulation results}
In the $4$-QAM case, the spectral efficiency of the Golden Reed-Solomon codes is given by
\begin{displaymath}
\frac{8k \text{ bits }}{2n \text{ channel uses }}=\frac{4k}{n} \text{ bpcu }
\end{displaymath}
From Lemma \ref{delta_d}, we get a lower bound for $\Delta_{\min}$: using an $(n,k,d_{\min})$ Reed-Solomon code, we have $\Delta_{\min}\geq \delta d_{\min}^2$. \\
If $k=\frac{n}{2}$, the spectral efficiency is $2$bpcu. Comparing the $4$-QAM, $(n,k,d_{\min})$ Golden-RS design ($E_{\mathcal{S}}=0.5$) with the uncoded Golden Code using BPSK ($E_{\mathcal{S},U}=0.25$), we get an asymptotic coding gain of:
\begin{equation} \label{gamma_2}
\gamma_{\as}=\frac{\sqrt{\Delta_{{\min}}}/E_{\mathcal{S}}}{\sqrt{\Delta_{{\min},U}}/E_{\mathcal{S},U}}=\frac{d_{\min}/0.5}{1/0.25}=\frac{d_{\min}}{2}
\end{equation}

Figures \ref{comparison} and \ref{634} show the performance comparisons of the Golden-RS codes $(4,2,3)$ and $(6,3,4)$ with the corresponding uncoded schemes at the spectral efficiency of $2$ bpcu.\\
Assuming the channel to be constant for $4$ blocks and $6$ blocks respectively, the Golden-RS codes outperform the uncoded scheme  by $6.1$ $\dB$ and $7.0$ $\dB$.\\
The gain for the $(4,2,3)$ code is unexpectedly high compared with the theoretical coding gain (\ref{gamma_2}) for $d=3$, that is $10 \log_{10}\left(\frac{3}{2}\right)\dB=1.7 \dB$. The rough estimate (\ref{gamma_2}) is based on the worst possible occurrence, that of a codeword of Hamming weight $3$ in which all three non-zero components correspond to invertible elements in the quotient. \\
However, we can verify empirically that in the $4$-QAM case and with our choice of the $(4,2,3)$ code, this event does not take place and in fact the actual value for $\Delta_{\min}$ found by computer search is $18$, giving an estimate for the gain of $3.2 \dB$,  a little closer to the observed value. \\
This favorable behavior might be due to the fact that the chosen constellation contains only one point in each coset, so that the codewords of Hamming distance $3$ are few.\\
Also for the $(6,3,4)$ code, the actual gain ($7.0 \dB$) is higher than the theoretical gain ($10\log_{10}2\dB=3.0 \dB$ based solely on the minimum Hamming distance; $5.3$ $\dB$ using the true value of $\Delta_{\min}$, that is $46$.)
\begin{figure}[bt] 
\begin{center}
\includegraphics[width=0.4\textwidth]{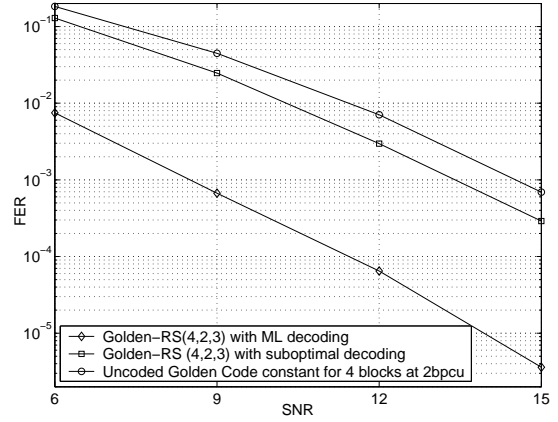}
\caption{Comparison between suboptimal decoding and ML decoding for the RS$(4,2,3)$ code at $2$ bpcu. The first method achieves a gain of only $1.1 \dB$ over the uncoded case, compared to the $6.1 \dB$ of the second.}
\label{comparison}
\end{center}
\end{figure}

\begin{figure}[bt]
\begin{center}
\includegraphics[width=0.4\textwidth]{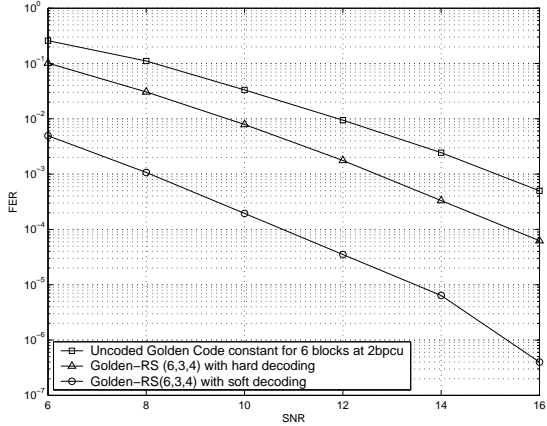}
\caption{Comparison between suboptimal decoding and ML decoding for the RS$(6,3,4)$ code at $2$ bpcu. The first method achieves a gain of $2.4 \dB$ over the uncoded case, compared to the $7.0 \dB$ of ML decoding.}
\label{634}
\end{center}
\end{figure}

\subsection{Sub-optimal decoding}
One can replace ML decoding with $n$ separate Sphere Decoders on each of the $n$ components of $\mathbf{Y}$. The signal is then demodulated, and mapped to a vector $(\hat{V}_1,\ldots,\hat{V}_n)$ in $\F_{256}^n$ using the inverse mappings of Steps \ref{4} and \ref{3} in Section \ref{4Q}. The received sequence $(\hat{V}_1,\ldots,\hat{V}_n)$ doesn't necessarily belong to the RS code, so a final step of RS decoding is needed.
\medskip \par
This \vv{hard} decoding has the advantage of speed and allows to use longer Reed-Solomon codes with high minimum distance. However it is highly suboptimal;
performance simulations show that with this method the coding gain is almost entirely cancelled out (see figure \ref{comparison}). \\
Suboptimal decoding also provides a good initial bound of the distance of the received point from the lattice, which can be used as a cost function for the stack decoder described in Section \ref{stack}. 

\begin{figure}[btp] 
\begin{center}
\includegraphics[width=0.4\textwidth]{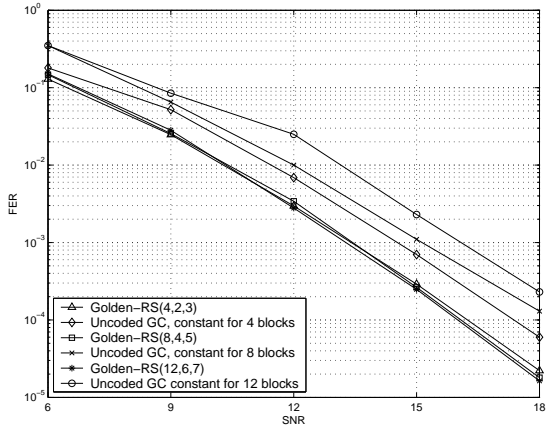}
\caption{Performance of $(4,2,3)$, $(8,4,5)$, and $(12,6,7)$ Golden Reed-Solomon codes with suboptimal decoding at $2$ bpcu compared with the uncoded Golden Code scheme with the same spectral efficiency.}
\label{2bpcu}
\end{center}
\end{figure}

\begin{figure}[btp] 
\begin{center}
\includegraphics[width=0.4\textwidth]{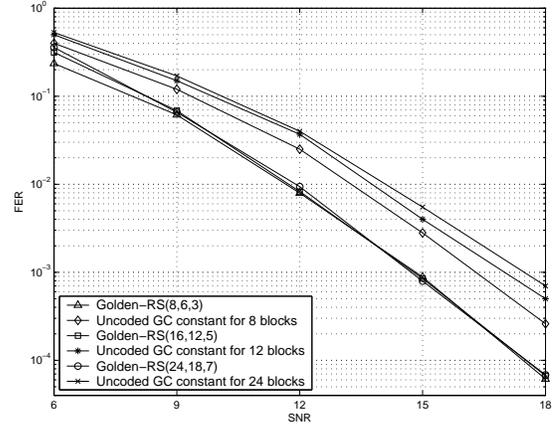}
\caption{Performance of $(8,6,3)$, $(16,12,5)$, and $(24,18,7)$ Golden Reed-Solomon codes with suboptimal decoding at $3$ bpcu compared with the uncoded Golden Code scheme with the same spectral efficiency.}
\label{3bpcu}
\end{center}
\end{figure}

\begin{itemize}
\item \textbf{2 bpcu:}\
Figure \ref{2bpcu} shows the performance comparison of the Golden-RS codes with suboptimal decoding with the uncoded scheme at the spectral efficiency of $2$ bpcu.\\
Assuming the channel to be constant for $4$, $8$ and $12$ blocks respectively, the $(4,2,3)$, $(8,4,5)$ and $(12,6,7)$ Golden-RS codes outperform the uncoded scheme at the same spectral efficiency by $1.1$ $\dB$, $1.7$ $\dB$ and $2.8$ $\dB$ at the FER of $10^{-3}$.\\
The Golden-RS schemes seem to be more robust on slow fading channels; in fact the performances of the Golden-RS$(n,k,d_{\min})$ codes on a channel which is constant for $n$ blocks remain almost unchanged (the variation is less than $0.2 \dB$) when $n$ varies between $4$ and $12$, while the uncoded Golden Code has a loss of almost $1.5$ $\dB$. 
\item \textbf{3 bpcu:}\
Assuming the channel to be constant for $8$, $16$ and $24$ blocks respectively, the $(8,6,3)$, $(16,12,5)$ and $(24,18,7)$ Golden-RS codes gain $1.5$ $\dB$, $2.2$ $\dB$ and $2.8$ $\dB$ over the uncoded scheme at the FER of $10^{-3}$ (see Figure \ref{3bpcu}).\\
Similarly to the previous case, the Golden-RS$(n,k,d_{\min})$ codes lose less than $0.3 \dB$ when $n$ varies between $8$ and $24$, while the Golden Code has a loss of $1.1 \dB$.
\end{itemize}

\subsection{The $16$-QAM case}
Using $16$-QAM modulation for each symbol $a,b,c,d$ in a Golden codeword, there are $2^{16}$ available Golden codewords, or $256$ words for each of the $256$ cosets of $2\gc$ in $\gc$. \\
As in the $4$-QAM case, we consider coset codes where the outer code is an $(n,k,d_{\min})$ Reed-Solomon code $\mathcal{C}$ on the quotient $\gc/2\gc$. Intuitively, the minimum distance of the Reed-Solomon code \vv{protects} the cosets from being decoded wrongly; if this choice is correct, the estimate for the right point in the coset is protected by the minimum determinant in $2\gc$. \\
The total information bits transmitted are $8k+8n$; they will be encoded into $8n+8n=16n$ bits.
\begin{itemize}
\item[-] The code $\mathcal{C}$ outputs $8n$ bits, which are used to encode the first two bits of $4n$ $16$-QAM constellations, that is the bits which identify one of the four cosets of $2\Z[i]$ in $\Z[i]$; each byte corresponds to a different coset configuration of $(a,b,c,d)$ (see Figure \ref{16QAM}). 
\item[-] the other $8n$ bits, left uncoded, are used to choose the last two bits of each $16$-QAM signal. 
\end{itemize}
In total, we have $4n$ $16$-QAM symbols, that is a vector of $n$ Golden codewords $\mathbf{X}=(X_1,\ldots,X_n)$. 
\begin{figure}[tbp]
\begin{center}
\includegraphics[width=0.35\textwidth]{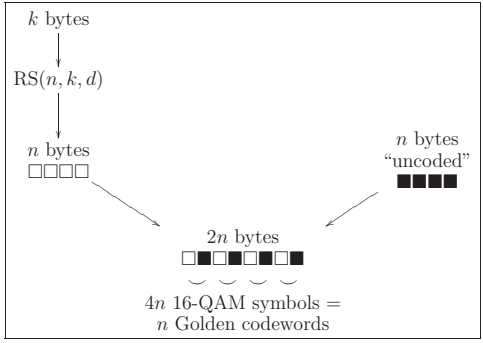}
\end{center}
\caption{The output of the Reed-Solomon code and the uncoded bits are \vv{mingled} before mo\-du\-la\-tion.}
\end{figure}
\begin{figure}[tbp]
\centering
\includegraphics[width=0.2\textwidth]{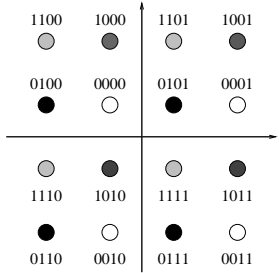}
\caption{The labelling of the $16$-QAM constellation used for performance simulations. The first and second bit identify one of the four cosets of $2\Z[i]$ in $\Z[i]$ (drawn in different shades of gray); the third and fourth bit identify one of the four points in the coset. We remark that this type of labelling cannot be a Gray mapping.}
\label{16QAM}
\end{figure}
The resulting spectral efficiency is
\begin{displaymath}
\frac{8(k+n) \text{ bits}}{2n \text{ channel uses}}=\frac{4(k+n)}{n} \text{ bpcu }
\end{displaymath}
In this case, the coding gain depends on the minimum Hamming distance inside each coset in addition to the minimum Hamming distance in the quotient: we have seen in (\ref{bound}) that 
\begin{equation} \label{d1d2}
\Delta_{\min} \geq \min\left(\min_{X \in 2\gc \setminus\{0\}},d_{\min}^2\right)=\min(16,d_{\min}^2)
\end{equation}

With an error-correcting code of rate $k=\frac{n}{2}$, we obtain a spectral efficiency of $6$ bpcu.
\begin{itemize}
\item[-]If $d_{\min}\geq 4$, we have $\gamma_{\as}=\frac{4/2.5}{1/1.5}=2.4$, leading to an approximate gain of $3.8 \dB$. Thus it does not seem worthwhile to use long codes with a high minimum distance with this scheme.
\item[-]If $d_{\min}=3$, $\gamma_{\as}=\frac{3/2.5}{1/1.5}=1.8$, making for a gain of $2.5 \dB$.
\end{itemize}

\subsection*{Decoding}
The ML decoding procedure for the $16$-QAM case requires only a slight mo\-di\-fi\-ca\-tion with respect to Step $6$ illustrated in Section \ref{decoding}. In the first phase, for each component $i=1,\ldots,n$ and for each coset leader $W_j$, $j=0,\ldots,255$, we find the closest point in that coset to the received component $Y_i$, that is
$$\hat{X}_{i,j}=\argmin_{X \in 2\gc} \norm{Y_i-H(X+W_j)}^2$$
Computing $HX$ and $HW_j$ separately allows to perform only $512$ products instead of $256^2$.
The second phase can be performed as in the $4$-QAM case, and the search is limited to the \vv{closest points} $\hat{X}_{i,j}+W_j$ determined in the previous phase:
$$ \hat{\mathbf{X}}=\argmin_{(\hat{X}_{1,j_1}+W_{j_1},\ldots,\hat{X}_{n,j_n}+W_{j_n})} \sum_{i=1}^n\norm{H(\hat{X}_{i,j_i}+W_{j_i})-Y_i}^2 $$
over all the images $(W_{j_1},\ldots,W_{j_n})$ of Reed-Solomon codewords.

\subsection*{Simulation results}
In the $16$-QAM case, the $(4,2,3)$ and $(6,3,4)$ Golden Reed-Solomon codes achieve a gain of $3.9 \dB$ and $4.3 \dB$ respectively over the uncoded scheme at $6$ bpcu at the frame error rate of $10^{-2}$, supposing that the channel is constant for $4$ and $6$ time blocks (see figure \ref{16QAM_new}).  

\begin{figure}[btp] 
\begin{center}
\includegraphics[width=0.4\textwidth]{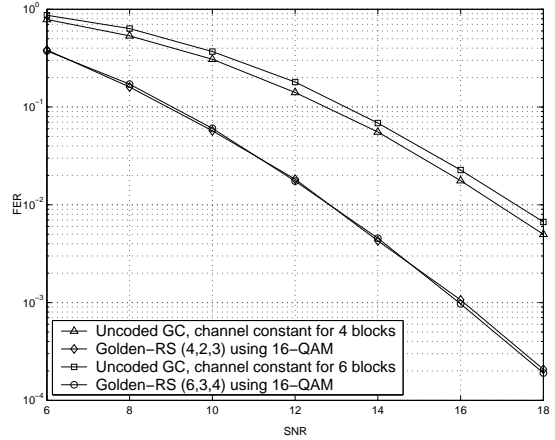}
\caption{Performance of the $(4,2,3)$ and $(6,3,4)$ Golden Reed-Solomon codes with ML decoding at $6$ bpcu compared with the uncoded schemes with the same spectral efficiency.}
\label{16QAM_new}
\end{center}
\end{figure}

\section{Conclusions}
In this paper we have presented Golden-RS codes, a coded modulation scheme for $2 \times 2$ slow fading MIMO channels, where the inner code is the Golden Code. \\
We use a simple binary partitioning, whose set of coset leaders coincides with a QAM symbol constellation. With a Reed-Solomon code as the outer code in order to increase the minimum Hamming distance among the codewords, we obtain a significant performance gain with respect to the uncoded case.

\appendices 

\section{Proofs} \label{proofs}
We report here some of the proofs for the results stated in the main part of the paper.
\begin{proof}[Proof of Lemma \ref{formula_Belfiore}]
For all $i=1,\ldots,L$, let $Q_i=X_i X_i^H$: then
\begin{align*}
&\Det(X_1 X_1^H + \ldots+ X_L X_L^H)\mathds{1}=\\
&=\Det(Q_1+\ldots+Q_L)\mathds{1}=\\ 
&=(\widetilde{Q}_1+\ldots+\widetilde{Q}_L)(Q_1+\ldots+Q_L)\mathds{1}=\\
&=\sum_{i,j=1}^L\widetilde{Q}_iQ_j= \sum_{i=1}^L\Det(Q_i)\mathds{1} + \sum_{i \neq j}\widetilde{Q}_iQ_j
\end{align*} 
We need to show that $\widetilde{Q}_iQ_j + \widetilde{Q}_jQ_i=\norm{\widetilde{X}_j X_i}_F^2 \mathbbm{1}$.\\
But $\norm{X}_F^2=\tr(XX^H)$, and therefore $\norm{\widetilde{X}_j X_i}_F^2= \tr(\widetilde{X}_jX_i X_i^H \widetilde{X}_j^H)$, and 
\begin{align*}
& \widetilde{Q}_jQ_i=\widetilde{X}_j^H \widetilde{X}_j X_i X_i^H, \quad \widetilde{Q}_iQ_j=\widetilde{\widetilde{Q}_jQ_i} \\
& \Rightarrow \widetilde{Q}_iQ_j + \widetilde{Q}_jQ_i= \tr(\widetilde{Q}_iQ_j)\mathbbm{1} = \tr(\widetilde{X}_jX_i X_i^H \widetilde{X}_j^H)\mathbbm{1},
\end{align*}
recalling that $\tr(AB)= \tr(BA)$.
\end{proof}

\begin{proof}[Proof of Remark \ref{summary}]
\begin{enumerate}
\item[a)] Let 
\begin{displaymath}
W=\left[\begin{array}{cc} w_1 & w_2 \\ i\overline{w}_2 & \overline{w}_1\end{array}\right],\quad w_1=t_1+is_1, w_2=t_2+is_2,
\end{displaymath}
where $t_1,t_2,s_1,s_2 \in \Z[\theta]$. Then $\norm{W}_F^2=\abs{w_1}^2+\abs{\overline{w}_1}^2+\abs{w_2}^2+\overline{\bar{w}_2}^2$. But $w_1=a+b\theta+i(c+d\theta)$ for some $a,b,c,d \in \Z$, and
\begin{multline*}
\abs{w_1}^2+\abs{\overline{w}_1}^2=\\
=(a+b\theta)^2+(c+d\theta)^2+(a+b\bt)^2+(c+d\bt)^2=\\ =2a^2+3b^2+2ab+2c^2+3d^2+2cd \in \Z
\end{multline*}
The same is true for $\abs{w_2}^2+\abs{\overline{w}_2}^2$.
\item[b)] If $X=\left[\begin{array}{cc} a & b \\ c & d\end{array}\right]$,
then 
\begin{multline*}
\fn{X}=\abs{a}^2+\abs{b}^2+\abs{c}^2+\abs{d}^2 \geq 2(\abs{ad}+\abs{bc}) \geq\\
 2\abs{ad-bc}=2\abs{\Det(X)}
\end{multline*}
and $$\fn{\widetilde{X}Y} \geq 2 \abs{\Det(\widetilde{X}Y)}=2\abs{\Det(X)\Det(Y)}$$ 
\item[c)] Let $X_1=\frac{1}{\sqrt{5}}A W_1, X_2=\frac{1}{\sqrt{5}}A W_2$, $W_1,W_2 \in \mathcal{O}$. Then
\begin{multline*}
\norm{\widetilde{X}_2X_1}_F^2=\frac{1}{25} \norm{\widetilde{W}_2 \widetilde{A} A W_1}_F^2=\\
\frac{\abs{N(\alpha)}^2}{25} \norm{\widetilde{W}_2 W_1}_F^2 = \frac{1}{5} \norm{\widetilde{W}_2 W_1}_F^2\geq \frac{2}{5},
\end{multline*}
since $W=\widetilde{W}_2 W_1$ belongs to $\mathcal{O}$.
\end{enumerate}
\end{proof}

\section{Quaternion Algebras} \label{qa}
This section summarizes some basic facts about quaternion algebras that are used in the paper. Our main references are the books of Vignéras \cite{Vig} and Reiner \cite{Re}.

\begin{defn}[\textbf{Quaternion algebras}] \label{quaternion_algebra}
Let $K$ be a field. A \emph{quaternion algebra} $\mh$ of center $K$ is a central simple algebra of dimension $4$ over $K$, such that there exists a separable quadratic extension $L$ of $K$, and an element $\gamma \in K^*$, such that
\begin{displaymath}
\mh=L\oplus L e, \qquad e^2=\gamma, \qquad ex=\sigma(x) e \quad \forall x \in L
\end{displaymath}
where $\sigma$ is the non-trivial $K$-automorphism of $L$. $L$ is called a \emph{maximal subfield} of $\mh$. $\mh$ will be denoted by the triple $(L/K, \sigma, \gamma)$.
\end{defn}
Quaternion algebras are a special case of \emph{cyclic algebras}. \\
To obtain a representation of $\mh$ as a $K$-module, consider a primitive element $i$ such that $L=K(i)$, and let $j=e$, $k=ij=j\sigma(i)$. Then
\begin{equation} \label{quaternions}
\mh=\left\{ a+bi+cj+dk \;|\; a,b,c,d \in K \right\}
\end{equation}
The following theorem gives a sufficient condition for a quaternion algebra to be a division ring:
\begin{theo} \label{norm}
Let $\mh=(L/K,\sigma,\gamma)$ be a quaternion algebra. If $\gamma$ is not a reduced norm of any element of $L$, then $\mh$ is a skewfield. 
\end{theo}

\begin{defn}[\textbf{Splitting fields}]
Let $\mh$ be a central simple $K$-algebra. An extension field $E$ of $K$ \emph{splits} $\mh$, or is a \emph{splitting field} for $\mh$, if 
$$E \otimes_K \mh \cong M_r(E)$$
\end{defn}

In the case of division algebras, every maximal subfield is a splitting field:
\begin{theo} \label{splitting_field} 
Let $\mathcal{D}$ be a skewfield with center $K$, with finite degree over $K$. Then every maximal subfield $E$ of $\mathcal{D}$ contains $K$, and is a splitting field for $\mathcal{D}$.
\end{theo}

In the following paragraphs we will always consider a Dedekind domain $R$, its quotient field $K$, and a quaternion algebra $\mh$ over $K$.

\begin{defn}[\textbf{Lattices and orders}] \label{lo}
A \emph{full $R$-lattice} or \emph{ideal} in $\mh$ is a finitely generated $R$-submodule $I$ in $\mh$ such that $K I=\mh$, where
\begin{displaymath}
KI=\left\{\sum_{i=1}^n k_i x_i \;\Big|\; k_i \in K,\, x_i \in I,\, n \in N\right\}
\end{displaymath} 
An \emph{$R$-order} $\Theta$ in $\mh$ is a full $R$-lattice which is also a subring of $\mh$ with the same unity element. A \emph{maximal $R$-order} is an order which is not properly contained in any other order of $\mh$. 
\end{defn}

For the following proposition see for example Reiner \cite{Re}:
\begin{prop}
A subring of $\mh$ containing a basis for $\mh$ over $K$ is an order if and only if all its elements are integral over $R$.
\end{prop}

\begin{rem}
The notion of order is a generalization of the notion of the ring of integers for commutative extensions. However, in the non-commutative case the set of elements which are integral over the base field might not be a ring.
\end{rem}

\begin{defn}[\textbf{Properties of ideals}] 
Given an ideal $I$ of $\mh$, we can define the \emph{left order} and the \emph{right order} of $I$ as follows:
\begin{align*}
& \Theta_{l}(I)=\{x \in \mh \,|\, Ix \subset I\}, \\
& \Theta_{r}(I)=\{x \in \mh \,|\, xI \subset I\}
\end{align*} 
$\Theta_{l}(I)$ and $\Theta_r(I)$ are orders. $I$ is called 
\begin{itemize}
\item \emph{two-sided} if $\Theta_l(I)=\Theta_r(I)$,  
\item \emph{normal} if $\Theta_l(I)$ and $\Theta_r(I)$ are maximal, 
\item \emph{integral} if $I \subset \Theta_l(I)$, $I \subset \Theta_r(I)$, 
\item \emph{principal} if $I=\Theta_l(I) x=x \Theta_r(I)$ for some $x \in \mh$ 
\end{itemize}
The \emph{inverse} of $I$ is the fractional ideal $I^{-1}=\{ x \in \mh \,|\, IxI \subset I\}$.\\
The \emph{norm} $N(I)$ of an ideal $I$ is the set of reduced norms of its elements, and it is an ideal of $R$. If $I=\Theta x$ is principal, $N(I)=R N(x)$.
\end{defn}

\section{Ideals, valuations and maximal orders} \label{ivmo}

\begin{defn}[\textbf{Prime ideals}]
Let $\Theta$ be an order, $\pg$ a two-sided ideal of $\Theta$ (that is, the left and right order of $I$ coincide with $\Theta$). $\pg$ is \emph{prime} if it is nonzero and $\forall I,J$ integer two-sided ideals of $\Theta$, $IJ \subset \pg \Rightarrow I \subset \pg$ or $J \subset \pg$.  
\end{defn}

The proofs of the following theorems can be found in Reiner's book \cite{Re}:

\begin{theo}
The two-sided ideals of an order $\Theta$ form a free group generated by the prime ideals.
\end{theo}

\begin{theo} \label{1-1}
Let $\Theta$ be a maximal order in a quaternion algebra $\mh$. Then the prime ideals of $\Theta$ coincide with the maximal two-sided ideals of $\Theta$, and there is a one-to-one correspondence between the prime ideals $\pg$ in $\mh$ and the prime ideals $P$ of $R$, given by $P=R \cap \pg$.\\
Moreover, $\Theta/\pg$ is a simple algebra over the finite field $R/P$.  
\end{theo}

\begin{defn}[\textbf{Valuations and local fields}]
A \emph{valuation} $v$ of $K$ is a positive real function of $K$ such that $\forall k, h \in K$,
\begin{enumerate}
\item $v(k)=0 \Leftrightarrow k=0$,
\item $v(kh)=v(k)v(h)$,
\item $v(k+h)\leq v(k)+v(h)$.
\end{enumerate}
$v$ is \emph{non-archimedean} if $v(k+h) \leq \max(v(k),v(h)) \; \forall k,h \in K$; it is \emph{discrete} if $v(K^*)$ is an infinite cyclic group.\\
$K$ can be endowed with a topology induced by $v$ in the following way: a neighborhood basis of a point $k$ is given by the sets 
\begin{displaymath}
U_{\varepsilon}(k)=\{h \in K \,|\, v(h-k)<\varepsilon\}
\end{displaymath}
$K$ will be called \emph{complete} if it is complete with respect to this topology.\\
If $v$ is non archimedean, the set 
\begin{displaymath}
R_v=\{k \in K \,|\, v(k) \leq 1\}
\end{displaymath}
is a local ring, called the \emph{valuation ring} of $v$. The quotient $R_v/P_v$, where $P_v$ is the unique maximal ideal of $R_v$, is called the \emph{field of residues} of $K$.\\ 
$K$ is a \emph{local field} if it is complete with respect to a discrete valuation $v$ and if $R_v/P_v$ is finite. 
\end{defn}

\begin{defn}[\textbf{Places}] 
A \emph{place} $v$ of $K$ is an immersion $i_v:K \to K_v$ into a local field $K_v$. If $v$ is non-archimedean, we say that it is a \emph{finite place}; otherwise, that it is an \emph{infinite place}.
\end{defn}

The finite places of $K$ arise from discrete $P$-adic valuations of $K$, where $P$ ranges over the maximal ideals in the ring of integers $R$ of $K$. (Recall that the ring of integers in a number field is always a Dedekind domain, and so the maximal ideals coincide with the prime ideals).
\medskip \par

\begin{defn}[\textbf{Ramified places}]
Let $\mh$ be a quaternion algebra over $K$, and $P$ a place of $K$.\\
Consider the $K$-module $\mh_P=\mh \otimes_K K_P$; $\mh_P$ is isomorphic to a matrix algebra $M_r(D)$ over a skew field $D$ of center $K_P$ and index $m_P$ over $K_P$; $m_P$ is called the \emph{local index} of $\mh$ at $P$.  
We say that $P$ is \emph{ramified} in $\mh$ if $m_P>1$. 
\end{defn}

Given a maximal order $\Theta$, the set $\Ram(\mh)$ of ramified places of $\mh$ is related to a particular two-sided ideal of $\Theta$:

\begin{defn}[\textbf{Different and discriminant}]
Let $\Theta$ be an order. The set 
\begin{displaymath}
\Theta^*=\{x \in \mh \,|\, tr(x \Theta) \subset R\}
\end{displaymath}
is a two-sided ideal, called the \emph{dual} of $\Theta$. Its inverse $\mathfrak{D}=(\Theta^*)^{-1}$ is a two-sided integral ideal, called the \emph{different} of $\Theta$.  
If $\{w_1,\ldots,w_4\}$ is a basis of $\Theta$ as a free $R$-module, 
\begin{displaymath}
(n(\mathfrak{D}))^2=R \Det(\tr(w_i w_j))
\end{displaymath}
The ideal $n(\mathfrak{D})$ of $R$ is called the \emph{reduced discriminant} of $\Theta$ and is denoted by $d(\Theta)$.
\end{defn}

\begin{prop} 
If $\Theta, \Theta'$ are two orders and $\Theta' \subsetneq \Theta$, then $d(\Theta') \subsetneq d(\Theta)$.  
\end{prop}

The notion of ramification for quaternion algebras is a generalization of the notion of ramification for field extensions:

\begin{theo}
Let $\Theta$ be a maximal order in $\mh$. For each place $P$ of $K$, let $m_P$ be the local index of $\mh$ at $P$, and let $\pg$ be the prime ideal of $\Theta$ corresponding to $P$ (see Theorem \ref{1-1}). Then $m_P>1$ only for a finite number of places $P$, and
\begin{displaymath}
P\Theta=\pg^{m_P}, \qquad \mathfrak{D}=\prod_{P \in \Ram(\mh)} \pg^{m_P-1}
\end{displaymath}
\end{theo}

\begin{prop} \label{discriminant} 
Let $\mh$ be a quaternion algebra unramified at infinity.\\
A necessary and sufficient condition for an order $\Theta$ to be maximal is that 
\begin{displaymath}
d(\Theta)=\prod_{P \in \Ram(\mh) \setminus \infty} P
\end{displaymath}
\end{prop}

In the case of infinite places $P$, the $P$-adic completion can be $\R$ (\emph{real primes}) or $\C$ (\emph{complex primes}).
Complex primes are never ramified \cite{Re}.

\begin{theo} \label{two-sided}
The two-sided ideals of a maximal order $\Theta$ form a commutative group with respect to multiplication, which is generated by the ideals of $R$ and the ideals of reduced norm $P$, where $P$ varies over the prime ideals of $R$ that are ramified in $\mh$.
\end{theo}

\bibliographystyle{IEEE}

\end{document}